\documentclass[11pt,a4paper,fleqn]{article}
\usepackage{amsmath,amssymb,amsthm,enumerate}
\usepackage{amsmath,amsthm,mathtools}
\usepackage[mathscr]{eucal}
\usepackage{hyperref}
\hypersetup{colorlinks, linkcolor=blue, citecolor=red, urlcolor=magenta}

\allowdisplaybreaks
\newcommand{\const}{\mathop{\rm const}\nolimits}

\newcommand{\todo}[1][\null]{\ensuremath{\clubsuit}}

\newcommand{\noprint}[1]{}
\newcommand{\checked}[1][\null]{\ensuremath{\boldsymbol{\surd}}}

\newtheorem{theorem}{Theorem}

\newtheorem{proposition}[theorem]{Proposition}
\newtheorem*{proposition*}{Proposition}
{\theoremstyle{definition}

\newtheorem*{notation*}{Notation}
}

\begin{document}
\par\noindent {\LARGE\bf
Light-like parallel vector fields and Einstein equations of gravity

\par}

\vspace{5mm}\par\noindent{\large
Isidore Mahara, C\'elestin Kurujyibwami and Venuste Nyagahakwa $^{\dag}$
}\par\vspace{2mm}\par

\vspace{2mm}\par\noindent{\it
$^\dag$\,University of Rwanda, College of Science and Technology, P.O.\,Box: 3900, Kigali, Rwanda}\\
$\phantom{^\dag}$\,E-mails: maraha03@yahoo.fr, celeku@yahoo.fr and venustino2005@yahoo.fr

\vspace{6mm}\par\noindent\hspace*{10mm}\parbox{140mm}{\small
We prove that, contrary to the situation with time-like and space-like parallel vector fields, there are real gravitational fields satisfying Einsteins equations of gravity and admitting nontrivial light-like parallel vector fields; we solve completely the field equations under this geometric constraint and obtain through Bianchi identities an interesting property for the Riemann curvature tensor corresponding with a real gravitational field admitting a nontrivial parallel vector field. The solutions we obtained turn out to be generically wave-like and then we prove by a geometric method that this class of solutions is quite different from the class of solutions corresponding with spherical waves.

\vspace{2mm}\par\noindent
}

\looseness=-1

\vspace{4mm}

\section{Introduction}\label{intro}
One of the authors of the present paper had already considered the problem of parallel transport in a gravitational field in~\cite{Ma}. Actually, as well stressed in~\cite{We} it is always possible to define on a (pseudo-)Riemannian manifold a parallel displacement which, at the infinitesimal level, has the same properties as the parallel displacement in linear spaces; but the integrability of the process of parallel transport for any vector is an intrinsic property of linear spaces. It is then an interesting mathematical problem to try to find for which kind of vectors the parallel transport on a (pseudo-)Riemannian manifold is an integrable process.

The results in this paper extend the work done in~\cite{Ma}, where is proved the non-existence of non-trivial time-like or space-like parallel vector fields in a real gravitational field satisfying Einstein's equations of gravity. At variance there are non-trivial gravitational fields admitting light-like parallel vector fields and the solutions to the corresponding field equations are generically wave like. Of course, wave-like solutions to Einstein's equations can be found through Physics literature and we will give some examples of them. However, when possible, it is always interesting to relate a class of solutions to a geometric property or a symmetry group of the underlying differentiable manifold.  It is in this spirit that we use geometric arguments to prove that the class of solutions we have obtained through the existence of a nontrivial parallel vector field and the class of spherical waves are quite different. It is also in this spirit that we considered in~\cite{Che} the covariant formulation of the strong equivalence principle and deduced through Bianchi identities an interesting property of the Riemann curvature tensor. Generally, geometric considerations are useful and, even from the above mentioned negative result about time-like and space like vector fields, we can deduce the following well-known result: \emph{A real gravitational field can never be stationary in a synchronous reference frame}. In fact, if the metric tensor 
$ds^2\equiv (dx^{0})^2+g_{\alpha \beta} dx^{\alpha} dx^{\beta}$ does not depend on $x^{0}$ then the time-like vector fields $\frac{\partial}{\partial x^0}$ is parallel and the corresponding gravitational field is necessarily trivial.

The structure of this paper is the following: 
In Section~2 we describe the properties of parallel vector fields on (pseudo-)Riemannian manifolds. In Section~3 we describe Einstein's equations and prove a statement about the energy-momentum tensor of a gravitational field admitting a nontrivial parallel vector field.
The Section~4 deals with Einstein's field equations in a coordinate system adapted to a parallel vector field. Actually, in this case, the field equations amount to the annihilation of the Ricci tensor. In Section~5 we use geometric considerations about the Riemann curvature tensor and solve completely the field equations.
 In Section~6 we prove that the class of solutions we have obtained and the class of solutions corresponding to spherical waves are essentially different. 
Section~7 provides examples of real gravitational fields admitting nontrivial parallel vector fields. We end this paper by giving concluding remarks. 

This paper uses standard results of differential geometry as presented in~\cite{St} and deals only with classical general relativity as presented in~\cite{La}.

\section{Parallel vector fields on (pseudo-)Riemannian manifolds}
Let $M$ be an $n$-dimensional manifold with metric tensor~$g$ given by its components~$g_{ij}$ in a local coordinate system~$(x^i), 1\leq i\leq n$.

A vector field $X$ defined by its contravariant components $(X^i)$ or its covariant components $(X_i)$  is called parallel (or covariantly constant) if it is parallel with respect to any piecewise smooth curve on $M$.
Then $X$ will be parallel if it satisfies the following equations
\begin{equation}\label{Equation21}
X^{i}_{;j}\equiv \frac{\partial X^i}{\partial x^j}+\Gamma^i_{jk}X^k=0,
 \end{equation}
where $\Gamma^i_{jk}$ are the Christoffel symbols associated with $g$ in the coordinate system $ (x^i)$. The covariant form of \eqref{Equation21} reads to
\begin{equation*}\label{Equation22}
X_{i;j}\equiv \frac{\partial X_i}{\partial x^j}-\Gamma^k_{ij}X_k=0.
\end{equation*}
From these equations it follows that
\begin{equation*}\label{Equation23}
\frac{\partial X_i}{\partial x^j}-\frac{\partial X_j}{\partial x^i}=0,
\end{equation*}
and then, by Poincar\'e theorem we conclude that: \textit{Any parallel vector field is a gradient field}.

From \eqref{Equation21}, it is easy to deduce that a necessary condition for $X$ to be parallel is that
\begin{equation}\label{Equation24}
R_{ijkl}X^l=0.
\end{equation}
As a consequence of Eq. \eqref{Equation24}, we have also the relations
\begin{equation}\label{Equation25}
R_{jl}X^l=0.
\end{equation}

In (\ref{Equation24}) and (\ref{Equation25}) $R_{ijkl}$ and $R_{kl}$ are  respectively the components of the Riemann curvature tensor and the Ricci tensor associated with $g$.

\section{Einstein equations of gravity, parallel vector fields and the energy-momentum tensor}
In general relativity, gravitational fields are governed by Einstein field equations
\begin{equation}\label{Equation31}
R_{ij}-\frac{1}{2}Rg_{ij}= \kappa T_{ij},
\end{equation}
where $g_{ij}$ is the metric tensor of space-time in a coordinate system $(x^{i})$; $R_{ij}$ and $R$ are respectively the corresponding Ricci tensor and Riemann scalar curvature; $\kappa$ is the gravitational constant whose value depends on the choice of the units system; $T_{ij}\equiv (p+\varepsilon)u_iu_j-pg_{ij}$ is the energy momentum tensor of the macroscopic body creating the gravitational field with $p,$  $\varepsilon$ and $u_i$  the pressure, the proper energy and  the $4$-velocity, respectively.  Actually, such a $T_{ij}$ is the energy-momentum tensor of a perfect fluid, but, in many cases, the error resulting  from the identification of the macroscopic body with a perfect fluid is very small.The signature of the metric tensor $g_{ij}$ is $(1,3)$ or, equivalently, $(+, -, -, -)$. This signature corresponds with the relation $g_{ij} u^iu^j=1$ for the 4-velocity. 

The following proposition is the main statement of this section.
\begin{proposition}\label{Prop1}
If the metric $g_{ij}$ satisfies Einstein equations \eqref{Equation31} and admits a non-trivial parallel vector field $X$, then energy- momentum tensor $T_{ij}$ is identically zero. 
\end{proposition}

\begin{proof} A double contraction on both sides of~\eqref{Equation31} yields
\begin{equation}\label{Equation32}
R= - \kappa (\varepsilon -3p).
\end{equation}
Replacing in Eq. \eqref{Equation31} $R$ by  its value from Eq. \eqref{Equation32}, we get
\begin{equation}\label{Equation33}
R_{ij}+\frac{1}{2} \kappa (\varepsilon -p) g_{ij}= \kappa (\varepsilon +p)u_iu_j.
\end{equation}
From Eq. \eqref{Equation33} it follows that
\begin{equation}\label{Equation34}
R_{ij}u^iX^j+\frac{1}{2} \kappa (\varepsilon -p) g_{ij}u^iX^j= \kappa (\varepsilon +p)u_iu_ju^iX^j.
\end{equation}
Then Eq.~\eqref{Equation25} and the identity $u_iu^i=1$ imply
\begin{equation}\label{Equation35}
-\frac{1}{2} \kappa ( \varepsilon +3p)u_iX^i=0.
\end{equation}
Eq.~\eqref{Equation35} admits two solutions: $u_iX^i=0$ and $\varepsilon +3p=0$.
The solution $\varepsilon +3p=0$ implies necessarily $ \varepsilon =p=0 $ and the energy-momentum tensor is zero.

If $u_iX^i=0$, Eqs. \eqref{Equation25} and \eqref{Equation33} yield $(\varepsilon - p)g_{ij}X^j=0$. 
Since $g_{ij}$ is regular and $X$ is supposed to be different from zero, we have necessarily
$\varepsilon -p=0$. In general $\varepsilon -3p \geq0$ and then for this case also we have $\varepsilon=p=0$, 
which implies that the energy-momentum tensor is zero.
\end{proof}

 \section{Einstein equations in a coordinate system adapted to a parallel vector field}
Since there exists an open set of space-time where the nontrivial parallel vector field $X$ never vanishes, we can assume that $X$ is $\dfrac{\partial }{\partial x^0}$ in the local coordinate system $(x^0, x^1, x^2, x^3)$~\cite{St}. This means that
\begin{equation}\label{Equation41}
X=\frac{\partial }{\partial x^0},\quad g(X,X)=g_{00}=0.
\end{equation}
From Eq.~\eqref{Equation41} it follows that $X^0=1$, $X^\alpha=0$, $1\leq \alpha \leq 3$, and Eqs.~\eqref{Equation21} for the parallel vector fields amount to
\begin{equation*}\label{Equation42}
\Gamma^i_{jk}X^k \equiv \Gamma^i_{j0}=0.
\end{equation*}
This is the same as writing that  
\begin{equation}\label{Equation43}
g^{ik}\left(\frac{\partial g_{jk}}{\partial x^0}+\frac{\partial g_{0k}}{\partial x^j}-\frac{\partial g_{j0}}{\partial x^k}\right)=0.
\end{equation}
The system \eqref{Equation43} is equivalent to
\begin{equation}\label{Equation44}
\frac{\partial g_{jk}}{\partial x^0}+\frac{\partial g_{0k}}{\partial x^j}-\frac{\partial g_{j0}}{\partial x^k}=0.
\end{equation}
Since $X$ is a gradient field and $X_i=g_{ij}X^j=g_{i0}$ there exists a differentiable function $f$ such that
$g_{i0}=\frac{\partial f}{\partial x^{i}}$ with $\frac{\partial f}{\partial x^0}=g_{00}=0$.

Therefore from~\eqref{Equation44} it follows that
\begin{equation}\label{Equation45}
\frac{\partial g_{jk}}{\partial x^0}=0,
\end{equation}
which means that the metric tensor $g_{ij}$ depends only on the variables $x^1, x^2, x^3$.
From $\dfrac{\partial f}{\partial x^0}=0$ it follows that$\left(\dfrac{\partial f}{\partial x^1}\right)^2+\left(\dfrac{\partial f}{\partial x^2}\right)^2+\left(\dfrac{\partial f}{\partial x^3}\right)^2\neq0$ and we can choose $(x^1, x^2, x^3)$ such that the metric tensor takes the form
\begin{equation}\label{Equation46}
ds^2=2dx^0dx^1+g_{\alpha\beta}dx^{\alpha}dx^{\beta},
\end{equation}
where the indices $\alpha$, $\beta$ run through $1,2,3$ and the $g_{\alpha \beta}$'s do not depend on $x^0$.  
It is worth noting the following formulae
\begin{gather}\label{Equation47}
\det(g_{ij})=g_{23}^2-g_{22}g_{33}<0,\, g^{01}=1,\, g^{1\alpha}=0,\,\alpha = 1,2,3  \\ \label{Equation49}
g^{22}=-\frac{g_{33}}{\det(g_{ij})},\, g^{33}=-\frac{g_{22}}{\det(g_{ij})},\, g^{23}=\frac{g_{23}}{\det(g_{ij})}.
\end{gather}
In empty space Einstein equations read 
\begin{equation}\label{Equation411}
R_{ij}\equiv g^{kl}R_{kilj}=0.
\end{equation}
From \eqref{Equation24} it follows that $R_{ijk0}=0$ and using~\eqref{Equation47} and~\eqref{Equation49} in the relations \eqref{Equation411} leads to
\begin{gather}\label{Equation412}
R_{1223}=R_{1323}=0,\\ \label{Equation413}
g_{33}R_{1212}-2g_{23}R_{1213}+g_{22}R_{1313}=0,\\ \label{Equation414}
R_{2323}=0.
\end{gather}
The existence of non trivial solutions is due to the fact that~\eqref{Equation413} does not imply the relations~$R_{1212}=R_{1312}=R_{1313}=0$.

\section{Field equations' solutions}\label{sect5}
Hereafter, we use a standard coordinate system in order to solve the field equations.
From the relations~\eqref{Equation47}--\eqref{Equation49} and~\eqref{Equation414} it follows that the matrix 
\begin{equation}\label{Equation51}
(\gamma_{\alpha\beta})=-\left(\begin{matrix}
g_{22}&g_{23}\\
g_{23}&g_{33}
\end{matrix}\right)
\end{equation}
can be considered as the metric tensor of a flat $2$-dimensional Riemannian manifold for each fixed $x^1$. 
Therefore, there exists a coordinate system where the $\gamma_{\alpha \beta}$  depend only on $x^1$. 
It is even possible to choose the coordinate system such that, in the metric \eqref{Equation46}, $g_{22}=g_{33}=-1$, $g_{23}=0$.
In such a coordinate system the components of the Riemann curvature tensor in Eqs. \eqref{Equation412} to \eqref{Equation414} are as follows:
\begin{equation}\label{Equation52}
R_{2323}= 0,
\end{equation}
\begin{equation}\label{Equation53}
R_{1223}=\frac{1}{2}\left(\frac{\partial ^2 g_{13}}{\partial x^2\partial x^2}-\frac{\partial ^2 g_{12}}{\partial x^2\partial x^3}\right),
\end{equation}
\begin{equation}\label{Equation54}
R_{1323}=\frac{1}{2}\left(\frac{\partial ^2 g_{13}}{\partial x^2\partial x^3}-\frac{\partial ^2 g_{12}}{\partial x^3\partial x^3}\right),
\end{equation}
\begin{equation}\label{Equation55}
R_{1212}=\frac{1}{2}\left(2\frac{\partial ^2 g_{12}}{\partial x^1\partial x^2}-\frac{\partial ^2 g_{11}}{\partial x^2\partial x^2}\right)-\frac{1}{4}\left(\frac{\partial g_{13}}{\partial x^2}-\frac{\partial g_{12}}{\partial x^3}\right)^2,
\end{equation}
\begin{equation}\label{Equation56}
R_{1313}=\frac{1}{2}\left(2\frac{\partial ^2 g_{13}}{\partial x^1\partial x^3}-\frac{\partial ^2 g_{11}}{\partial x^3\partial x^3}\right)-\frac{1}{4}\left(\frac{\partial g_{13}}{\partial x^2}-\frac{\partial g_{12}}{\partial x^3}\right)^2,
\end{equation}
\begin{equation}\label{Equation57}
R_{1213}=\frac{1}{2}\left(\frac{\partial ^2 g_{13}}{\partial x^1\partial x^2}+\frac{\partial ^2 g_{12}}{\partial x^1\partial x^3}-\frac{\partial ^2 g_{11}}{\partial x^2\partial x^3}\right),
\end{equation}
and Einstein field Eqs. \eqref{Equation412}--\eqref{Equation414} read
\begin{equation}\label{Equation58}
\frac{\partial ^2 g_{13}}{\partial x^2\partial x^2}-\frac{\partial ^2 g_{12}}{\partial x^2\partial x^3}=0,
\end{equation}
\begin{equation}\label{Equation59}
\frac{\partial ^2 g_{13}}{\partial x^2\partial x^3}-\frac{\partial ^2 g_{12}}{\partial x^3\partial x^3}=0,
\end{equation}
\begin{equation}\label{Equation510}
\frac{\partial ^2 g_{11}}{\partial x^2\partial x^2}+\frac{\partial ^2 g_{11}}{\partial  x^3\partial x^3}-2\left(\frac{\partial ^2 g_{12}}{\partial x^1\partial x^2}+\frac{\partial ^2 g_{13}}{\partial x^1\partial x^3}\right)+\left(\frac{\partial g_{13}}{\partial x^2}-\frac{\partial g_{12}}{\partial x^3}\right)^2=0.
\end{equation}
From \eqref{Equation58} and \eqref{Equation59} it follows that $\dfrac{\partial g_{13}}{\partial x^2}-\dfrac{\partial g_{12}}{\partial x^3}$ depends only on $x^1$. Therefore there exists a function $ \varphi $ such that $\dfrac{\partial g_{13}}{\partial x^2}-\dfrac{\partial g_{12}}{\partial x^3}=\varphi(x^1)$.

If $\varphi(x^1)\equiv0 $ there exists a function $f$ such that $g_{12}=\dfrac{\partial f}{\partial x^2}$,\, 
$g_{13}=\dfrac{\partial f}{\partial x^3} $. Then \eqref{Equation510} reads
\begin{equation}\label{Equation511} \frac{\partial ^2}{\partial x^2\partial x^2}\left(g_{11}-2\dfrac{\partial f}{\partial x^1}\right)+\frac{\partial ^2}{\partial x^3\partial x^3}\left(g_{11}-2\dfrac{\partial f}{\partial x^1}\right)=0.
\end{equation}
The general solution of the system \eqref{Equation58}--\eqref{Equation510} is given by
\begin{equation}\label{Equation512}
 g_{12}=\dfrac{\partial f}{\partial x^2},\quad  g_{13}=\dfrac{\partial f}{\partial x^3},
\quad g_{11}=2\dfrac{\partial f}{\partial x^1}+h(x^1,x^2,x^3)
\end{equation}
with $f$ an arbitrary differentiable function and $h$ any function which is harmonic with respect to $x^2$ and $x^3$.

If $ \varphi(x^1)\neq0 $ we can write $ g_{12}=\dfrac{\partial f}{\partial x^2} $. From $\dfrac{\partial g_{13}}{\partial x^2}-\dfrac{\partial g_{12}}{\partial x^3}=\varphi(x^1) $ it follows that $ g_{13}$ takes the form $g_{13}=\dfrac{\partial f}{\partial x^3}  +  x^2\varphi(x^1) +  \dfrac{\partial \psi(x^1, x^3)}{\partial x^3} $. However, since we do not change the value of $ g_{12} $ if we replace the arbitrary function $ f $ by $ f+\psi $ we can write $ g_{13}= \dfrac{\partial f}{\partial x^3} +x^2\varphi(x^1)$ and \eqref{Equation510} becomes
\begin{equation}\label{Equation513} \frac{\partial ^2}{\partial x^2\partial x^2}\left(g_{11}-2\dfrac{\partial f}{\partial x^1}\right)+\frac{\partial ^2}{\partial x^3\partial x^3}\left(g_{11}-2\dfrac{\partial f}{\partial x^1}\right)+\varphi^2(x^1)=0.
\end{equation}

The general solution of the system \eqref{Equation58}--\eqref{Equation510} is given by
 \begin{equation}\label{Equation5114}
 g_{12}=\dfrac{\partial f}{\partial x^2},\quad  g_{13}=\dfrac{\partial f}{\partial x^3}+ x^2\varphi(x^1),\quad g_{11}=2\dfrac{\partial f}{\partial x^1}+h(x^1,x^2,x^3)-\frac{1}{2}(x^2)^2\varphi^2(x^1)
\end{equation} 
with $f$ an arbitrary function and $h$ harmonic with respect to $x^2$ and $x^3$.

Formula \eqref{Equation5114}, together with  $g_{00}=g_{02}=g_{03}=g_{23}=0$, $g_{01}=1$, $g_{22}=g_{33}=-1$ determine completely the general solution of the field equations for a gravitational field admitting a nontrivial parallel vector field. 

It is worth noting that the coordinate system $(x^0, x^1, x^2, x^3)$ used to solve the field equations is related to the "natural" coordinate system $(x^{'0}, x^{'1}, x^{'2}, x^{'3})$ by the transformations 
\[
x^0=\frac{x^{'0}+x^{'1}}{\sqrt{2}},\, x^0=\frac{x^{'0}-x^{'1}}{\sqrt{2}},\, x^2=x^{'2},\,  x^3=x^{'3}.
\]
It is then clear that, generically, the solutions we have obtained are wave-like. However, the system \eqref{Equation58}--\eqref{Equation510} admits notrivial solutions which do not depend on $x^{1}$. Those solutions are not wave-like and in that case $\frac{\partial}{\partial x^{0}}$ and $\frac{\partial}{\partial x^{1}}$ are Killing vectors \cite{To}.

However, if in addition, $ \dfrac{\partial }{\partial x^1}$ is parallel, it follows from \eqref{Equation52}--\eqref{Equation57} that \eqref{Equation46} is the metric of a flat space-time. Furthermore, from the fact that $ \dfrac{\partial g_{13}}{\partial x^2}-\dfrac{\partial g_{12}}{\partial x^3} $ depends only on $ x^1 $ and $R_{1212}+ R_{1313}=0 $ we deduce from \eqref{Equation56}--\eqref{Equation58} the relations
\begin{equation}\label{Equation515}
\dfrac{\partial R_{1212}}{\partial x^2} + \dfrac{\partial  R_{1312}}{\partial x^3}=0, \quad\dfrac{\partial R_{1213}}{\partial x^2} + \dfrac{\partial  R_{1313}}{\partial x^3}=0,
\end{equation}
\begin{equation}\label{Equation516}
\dfrac{\partial R_{1213}}{\partial x^2}=\dfrac{\partial  R_{1212}}{\partial x^3}, \quad\dfrac{\partial R_{1213}}{\partial x^3}=-\dfrac{\partial R_{1212}}{\partial x^2}.
\end{equation}
The relations \eqref{Equation516} are the Cauchy-Riemann equations for the functions  $ R_{1212} $ and $ R_{1213} $ and  therefore we have the following proposition.
\begin{proposition}\label{Prop31}
If a gravitational field admits a nontrivial parallel vector field then there exists a reference frame where the only nonzero components (up to indices permutations) $R_{1212}$, $R_{1313}$, $R_{1213}$ satisfy the relations $R_{1212}+R_{1313}=0$,  $R_{1213} + iR_{1212}=f(x^1,z) $ with $f$ a holomorphic function with respect to $z$ and $z=x^2 + ix^3$.
\end{proposition}

It is worth giving an interpretation to the relations~\eqref{Equation515}. In fact in~\cite{Ge} the author shows that a covariant necessary condition for a gravity theory to ensure the implementation of the strong equivalence principle is given by the equations 
\begin{align}\label{equation517}
D_i R^i_{jkl}=0
\end{align}
in regions where matter is absent. On the other hand, a contraction on $i$ and $m$ in Bianchi identities $D_m R^i_{jkl}+D_k R^i_{jlm}+D_l R^i_{jmk}=0$ amounts to~\eqref{equation517} identically in regions of space-time where the Ricci tensor is identically zero; such is the case in gravitational fields admitting a nontrivial parallel vector field within the context of the classical general relativity and Eqs.~\eqref{equation517} reduce to~\eqref{Equation515} under the constraint $ R_{1212}+R_{1313}=0$.

\section{Spherical gravitational waves}
In this Section we consider the class of Robinson and Trautman metrics which are solutions of Einstein equations of gravity and we prove that only solutions corresponding to trivial gravitational fields admit non-trivial parallel vector fields. In empty space such metrics take the form
\begin{equation}\label{Equation71}
ds^2=2d \rho d\sigma+\left(K-2H\rho-2\frac{m}{\rho}\right) d\sigma^2-\frac{\rho^2}{p^2}\left( d\xi^2+d\eta^2\right)
\end{equation}
where $m$ is a function of $\sigma$ only and $p$ is a function of $\rho, \xi, \eta$.
\begin{equation}\label{Equation72}
H=\frac{1}{p}\frac{\partial p}{\partial \sigma}
\end{equation}
and $K$ is the Gaussian curvature of the surface $\rho=1, \sigma=\const$,
\begin{equation}\label{Equation73}
K=p^2\left(\frac{\partial^2}{\partial \xi^2}+\frac{\partial^2}{\partial \eta^2}\right)\ln p
\end{equation}
and 
\begin{equation}\label{Equation74}
\frac{\partial^2 K}{\partial \xi^2}+\frac{\partial^2 K}{\partial \eta^2}=\frac{4}{p^2}\left( \frac{\partial}{\partial \sigma}-3H \right)m.
\end{equation}
With the notations $x^0=\rho$, $x^1=\sigma$, $x^2=\xi, x^3=\eta$ the non-zero components of the metric tensor are given by
\begin{equation}\label{Equation75}
g_{01}=1,\, g_{11}=K-2H\rho-2\frac{m}{\rho},\, g_{22}=g_{33}=-\frac{\rho^2}{p^2},\, g^{00}=-K+2H\rho+2\frac{m}{\rho},\, g^{01}=1
\end{equation}
\begin{equation}\label{Equation76}
g^{22}=g^{33}=-\frac{p^2}{\rho^2}.
\end{equation}

Up to indices permutation the only components of the curvature tensor which are not identically zero are:
\begin{equation*}
\begin{split}
& R_{0101}=2\frac{m}{\rho^3},\, R_{0112}=-\frac{1}{2\rho}\frac{\partial K}{\partial \xi},\,R_{0113}=-\frac{1}{2\rho}\frac{\partial K}{\partial \eta},\, R_{1202}=R_{1303}=-\frac{m}{\rho p^2}, \\
 &R_{1223}=\frac{\rho}{2p^2}\frac{\partial K}{\partial \eta},\, R_{1323}=\frac{1}{2p^2}\frac{\partial K}{\partial \xi}, R_{2323}=-m\frac{\rho}{p^4} \text{ and } R_{1212}, R_{1213}, R_{1313}
 \end{split}
 \end{equation*}
 whose expressions are too tedious to be represented here.

\begin{proposition}
If $(X^i)$ is a parallel vector field with respect to a metric of class (\ref{Equation71}) and if $X^1=0$ then $(X^i)$ is the trivial vector field.
\end{proposition}
\begin{proof}
As proved in \cite{Ma}, if the vector field is parallel, then it is light--like. On the other hand, if $X^1$ equals zero, then
 \[g_{ij}X^iX^j=-\frac{\rho^2}{p^2}\left[ (X^2)^2+(X^3)^2\right]\leq 0.\]
Therefore, for $(X^i)$ to be light-like we must have $X^2=X^3=0$. In order to prove that $X^0$ is also zero, let us consider the equation  $\frac{\partial X^i}{\partial x^j}+\Gamma_{jk}^i X^k=0$ of the parallel vector field for $i=j=2$.  Since $X^2=0$ the equation reduces to $\Gamma_{20}^2X^0=0$ with $\Gamma_{20}^2=\frac{1}{\rho}$. Since $\frac{1}{\rho} \neq 0,\, X^{0}=0$ and the vector field $(X^i)$ is zero.

To prove that any metric of the class (\ref{Equation71}) admitting a non-trivial parallel vector field represents a trivial gravitational field, we use the fact that, if $X^l$ and $R_{ijkl}$ represent respectively the components of a parallel vector field and the components of the Riemann curvature tensor, then $R_{ijkl}X^l=0$. 

Let us consider the cases  $(i,j,k)=(0,1,2)$ and $(i,j,k)=(0,1,3)$. The corresponding equations  $R_{ijkl}X^l=0$ reduce respectively to $R_{0121}X^1=0$ and $R_{0131}X^1=0$. According to Proposition \ref{Prop1}, for $(X^l)$ not to be trivial, $R_{0121}=\frac{1}{2\rho}\frac{\partial K}{\partial \xi}=0$ and $R_{0131}=\frac{1}{2\rho}\frac{\partial K}{\partial \eta}=0$. Therefore $\frac{\partial K}{\partial \xi}=\frac{\partial K}{\partial \eta}=0$ and the components $R_{0112}, R_{0113}, R_{1223}, R_{1323}$ of the curvature tensor are zero.

The cases $(i,j,k)=(1,2,2)$ and $(i,j,k)=(1,3,3)$ lead to the equations  $R_{122l}X^l=0$ and $R_{033l}X^l=0$ which reduce  respectively to $R_{0221}X^1=0$ and $R_{0331}X^1=0$ and the condition $X^1\neq 0$ for $(X^l)$ not to be trivial implies $R_{0212}=R_{0313}=-\frac{m}{\rho p^2}=0$. Therefore $m=0$ and the components $R_{0110}$, $R_{0212}$, $R_{0313}$, $R_{2323}$ of the curvature tensor are zero.

With $(i,j,k)=(1,2,2)$ and $(i,j,k)=(1,3,3)$ lead to the equations  $R_{122l}X^l=0$ and $R_{133l}X^l=0$ which reduce  respectively to $R_{1221}X^1=0$ and $R_{1331}X^1=0$ because $R_{1220}$, $R_{1223}$, $R_{1330}$ and $R_{1332}$~are equal to zero.

The condition $X^{1}\neq 0$ for $(X^l)$ for a non-trivial $(X^l)$ implies $R_{1212}=R_{1313}=0$ and it remains to prove that $R_{1213}=0$ for $(X^l)$ not to be trivial. At this aim, we consider the equation $R_{123l}X^l=0$ which , since $R_{1230}=R_{1232}=0$ reduces to $R_{1231}X^1=0$ and leads to the constraint $R_{1213}=0$.

Actually, a solution to the equations $R_{1212}=R_{1213}=R_{1313}=0$ is obtained by imposing the condition $\frac{\partial p}{\partial \sigma}=0$ in addition to the previous constraints $m=0, \frac{\partial K}{\partial \xi}=\frac{\partial K}{\partial \eta}=0$. Then $p$ is a function of $\xi$ and $\eta$ only, $H=0$,  $K=\const$, $p$ being a solution of Eq. \eqref{Equation74} and this ends the proof of our statement.
\end{proof}

\section{Examples of gravitational fields with nontrivial parallel vector fields}
\subsection{Metric of Asher Peres }
Let us consider the metric in \cite{Pe} of the form
\begin{equation}\label{Equation61}
ds^2=dt^2-dx^2-dy^2-dz^2-2f(x+t,y,z)(dx+dt)^2.
\end{equation}
The coordinate transformation $t=t'$, $x=x'-t'$,  $y=y'$, $z=z'$
reduces \eqref{Equation61} to the form
\begin{equation}\label{Equation62}
ds^2=2dtdx-(1+2f(x,y,z))dx^2-dy^2-dz^2.
\end{equation}
With the notations $t=x^0,$ $x=x^1,$ $y=x^2,$ $z=x^3,$ it is clear that the vector field $\dfrac{\partial }{\partial x^0}$ is covariantly constant. As for the Riemann curvature tensor, up to permutation of indices, the only components which are not identically zero are
\[R_{1212}=\dfrac{\partial ^2 f}{\partial y^2},\quad R_{1313}=\dfrac{\partial ^2 f}{\partial z^2}, \quad R_{1213}=\dfrac{\partial ^2 f}{\partial y \partial z}\]
and Einstein equations in vacuo are satisfied if $f$ is a harmonic function of $y$ and $z$ whatever its dependence on $(x+t)$.

\subsection{Generalization of the weak plane gravitational wave solution}
Let us consider the metric tensor $g_{ij}(x^0)$  where the determinant $|g_{\alpha \beta}|$ of the submatrix $g_{\alpha \beta}$ for $ \alpha, \beta$ running through $1,2,3$ is equal to zero. It can be shown~\cite{La} that, by a coordinate transformation, the metric can be brought to the form
\begin{equation}\label{Equation63}
ds^2=2d\eta dx^1+g_{ab}(\eta )dx^adx^b
\end{equation}
with $a$ and $b$ running through $2,3$. It is clear that the vector field $\dfrac{\partial}{\partial x^1}$ is covariantly constant with respect to the metrics of the form \eqref{Equation63} which constitute a special case of the metrics \eqref{Equation46}. The solutions of Einstein's equations in vacuo represent a generalization of the weak plane gravitational wave solution~\cite{La}. The families of metrics \eqref{Equation62} and \eqref{Equation63} have been investigated in details with regard to the gravitational wave memory effect in~\cite{Zh}.

\section*{Conclusion }
Within the framework of classical general relativity we proved that, contrary to the case of time-like and space--like parallel vector fields, there are nontrivial solutions of Einstein equations admitting nontrivial light-like parallel vector fields. However, the main aim of the paper was not proving the existence of such parallel vector fields, but using this geometric constraint to solve the field equations. The corresponding solutions turned out to be generically wave--like and, using geometric reasoning, we proved that the class of the solutions we obtained is quite different from the class of spherical waves. In addition, using the covariant formulation of the strong equivalence principle presented in~\cite{Ge} we deduced an interesting property of the Riemann curvature tensor. Of course as we have seen this formulation is trivial in our case since, within the context of classical general relativity, the existence of a nontrivial parallel vector field implies the annihilation of the Ricci tensor which turns the system $D_i R^i_{jkl}=0$ into an identity. 
However, the existence of a parallel vector field does not always imply the annihilation of the Ricci tensor in any covariant gravity theory. Therefore, the combination of the existence of a nontrivial parallel vector field and a nontrivial annihilation of the covariant divergence of the curvature tensor as constraints on the metric tensor of space time can be used efficiently in the search of solutions of the field equations in alternative theories of gravity such as in monograph~\cite{Che}.

\end{document}